\documentclass[aps,twocolumn,floatfix]{revtex4}

\usepackage{stmaryrd,amsmath,amsthm,amsfonts,amssymb,times,bbm,graphicx}

\newtheorem{theorem}{Theorem}

\newtheorem{lemma}[theorem]{Lemma}

\def\R{\mathcal{R}}

\def\RR{\mathbbm{R}}

\def\EE{\mathbbm{E}}

\DeclareMathOperator{\tr}{tr}

\begin{document}

\title{Note on sampling without replacing from a finite collection
of matrices}

\author{David Gross and Vincent Nesme}

\affiliation{
	Institute for Theoretical Physics, Leibniz University Hannover,
	30167 Hannover, Germany
}

\email{www.itp.uni-hannover.de/~davidg}


\begin{abstract} 
	This technical note supplies an affirmative answer to a question
	raised in a recent pre-print in the
	context of a ``matrix recovery'' problem. Assume one samples $m$
	Hermitian matrices $X_1, \dots, X_m$ with replacement from a finite
	collection. The deviation of the sum $X_1+\dots+X_m$ from its
	expected value in terms of the operator norm can be estimated by an
	``operator Chernoff-bound'' due to Ahlswede and Winter. The question
	arose whether the bounds obtained this way continue to hold if the
	matrices are sampled without replacement. We remark that a positive
	answer is implied by a classical argument by Hoeffding. Some
	consequences for the matrix recovery problem are sketched.
\end{abstract}


\maketitle

This is a technical comment on \cite{gross_recovering_2009}. While we
provide a (minimal) introduction, readers not familiar with
\cite{gross_recovering_2009} may find the present note hard to follow.

\subsection{Motivation}

The \emph{low-rank matrix recovery problem} 
\cite{recht_guaranteed_2007,
candes_exact_2009,candes_power_2009,
singer_uniqueness_2009,
keshavan_matrix_2009,
wright_robust_2009,
gross_quantum_2009,
recht_simpler_2009,
gross_recovering_2009,
candes_robust_2009}
is:
Reconstruct a low-rank matrix $\rho$
from $m$ randomly selected matrix elements. The more general
version introduced in 
\cite{gross_recovering_2009} reads:
Reconstruct $\rho$
from $m$ randomly selected expansion
coefficients with respect to any fixed matrix basis.

Let us consider what seems to be the most mundane aspect of the
problem: the way in which the $m$ coefficients are ``randomly
selected''. Assume we are dealing with an $n\times n$ matrix $\rho$.
The statement of the matrix recovery problem calls for us to sample
$m$ of the $n^2$ coefficients characterizing $\rho$ \emph{without
replacing}. This yields a random subset $\Omega$ consisting of $m$ of
the $n^2$ coefficients, from which the matrix $\rho$ is then to be
recovered.

Due to the requirement that the drawn coefficients be distinct, the
$m$ samples are not independent. Their dependency turns out to impede
the technical analysis of the recovery algorithms. In order to avoid this
complication, most authors chose to first
analyze a variant where the revealed coefficients are drawn
independently and then, in a second step, relate the modified question
to the original one. Two such proxies for sampling without replacement
have been discussed:

\emph{1. The Bernoulli model}
\cite{candes_exact_2009,candes_power_2009,
candes_robust_2009}. Here, each of the $n^2$ coefficients is assumed
to be known with probability $\frac{m}{n^2}$. Thus the number of
revealed coefficients is itself a random variable (with expectation
value $m$). The minor draw-back of this approach is that, with finite
probability, significantly more than $m$ coefficients will be
uncovered. These possible violations of the rules of the original
problem have to be factored in, when the success probability of the
algorithm is computed.

\emph{2. The i.i.d.\ approach}
\cite{gross_quantum_2009,gross_recovering_2009,recht_simpler_2009}.
The known coefficients are obtained by sampling $m$ times \emph{with}
replacement. The draw-back here is that, with fairly high probability,
some coefficients will be selected more than once. To understand why
this is undesirable, we need to recall some technical definitions from
\cite{gross_recovering_2009}. 

Let $A_1,\dots, A_m$ be random variables taking values in $[1,n^2]$.
For now, assume the $A_i$'s are distributed uniformly and
independently. Let $\{w_a\}_{a=1}^{n^2}$ be an orthonormal Hermitian
basis in the space of $n\times n$-matrices. 
A central object in the analysis is the \emph{sampling operator},
defined as
\begin{equation}\label{eqn:r}
	\R: \rho \mapsto 
	\frac{n^2}{m} \sum_{i=1}^m \tr(\rho w_{A_i}) \,w_{A_i}.
\end{equation}

If the $A_i$ are all distinct, then $\frac{m}{n^2}\R$ is a
projection operator. If, on the other hand, some basis elements occur
more than once, the spectrum of the sampling operator will be more
complicated. More importantly, the operator norm $\|\frac{m}{n^2}\R\|$
may become fairly large. The latter effect is undesirable, as the
logarithm of the operator norm appears as a multiplicative constant in
the final bound on the number of coefficients which need to be known
in order for the reconstruction process to be successful.

There seem to be three ways to cope with this problem. First, use the
worst-case estimate $\|\frac{m}{n^2}\R\|\leq m$ (done in Section~II.C
of \cite{gross_recovering_2009}). Second, use the fact that the
operator norm is very likely to be of order $O(\log n)$ (suggested at
the end of Section~II.C in \cite{gross_recovering_2009} and
implemented in later versions of 
\cite{recht_simpler_2009}). Third, prove that the arguments in
\cite{gross_recovering_2009} remain valid when the $A_i$'s are chosen
without replacement. Supplying such a proof is the purpose of the
present note.

Following earlier work \cite{candes_exact_2009,candes_power_2009},
Ref.~\cite{gross_recovering_2009} reduces the
analysis of the matrix recovery problem to the problem of controlling
the operator norm of various linear functions of $\R$ (c.f.\ Lemma~4
and Lemma~6 of \cite{gross_recovering_2009}). This, in turn, is done
by employing a large-deviation bound for the sum of independent
matrix-valued random variables, which was derived in
\cite{ahlswede_strong_2002}.  Below, we point out that in some
situations this bound remains valid when the random variables are not
independent, but represent sampling without replacing.

\subsection{Statement}

Let $C$ be a finite set. For $1\leq m\leq |C|$, let $X_i$ be a random
variable taking values in $C$ with uniform probability. We assume that
all the $X_i$ are independent, so that ${\bf X}=\langle X_1, \dots,
X_m\rangle$ is a $C^m$-valued random vector modeling sampling
\emph{with} replacement from $C$. Likewise, let ${\bf Y}=\langle Y_1,
\dots, Y_m\rangle$ be a random vector of $C$'s sampled uniformly
\emph{without} replacement.

We are mainly interested in the case where $C$ is a finite set of
Hermitian matrices with some additional properties: We assume the
set is centered $\EE[X_i]=0$ and that there are constants
$c,\sigma_0\in\RR$ bounding the operator norm $\|X_i\|\leq c$ and the
variance $\|\EE[X_i^2]\|\leq\sigma_0^2$ of the random variables. Then:

\begin{theorem}[Operator-Bernstein inequality]\label{thm:bernstein}
	With the definitions above, let
	$S_{\bf X}=\sum_{i=1}^m X_i$ and
	$S_{\bf Y}=\sum_{i=1}^m Y_i$. 
	Let $V=m \sigma_0^2$. Then for both $S=S_{\bf X}$ and $S=S_{\bf Y}$
	it holds that
	\begin{equation}\label{eqn:bernstein}
		\Pr\big[\|S\|>t\big]
		\leq 2 n \exp\left(-\frac{t^2}{4V}\right),
	\end{equation}
	for $t \leq 2 V/c$, and
	\begin{equation}\label{eqn:bernstein2}
		\Pr\big[\|S\|>t\big]
		\leq 2 n \exp\left(-\frac{t}{2c}\right),
	\end{equation}
	for larger values of $t$. 
\end{theorem}

The version involving  $S_{\bf X}$ has been proved in
\cite{gross_recovering_2009} as a minor variation of the
operator-Chernoff bound from \cite{ahlswede_strong_2002}. In the
proof, the failure probability is bounded from above in terms of the
``operator moment-generating function''
\begin{equation*}
	M_{\bf X}(\lambda)=\EE[\tr\exp(\lambda S_{\bf X})].
\end{equation*}
To establish the more general statement, it would be sufficient to
show that $M_{\bf Y}\leq M_{\bf X}$. In fact, this relation is
well-known to hold for real-valued random variables. One popular way
of proving it involves the notion of \emph{negative association}
\cite{joag-dev_negative_1983,dubhashi_concentration_2009}. Indeed, the
author of \cite{gross_recovering_2009} tried to generalize this
concept to the case of matrix-valued random variables, but failed to
overcome its apparent dependency on the \emph{total} order
of the
real numbers. However, he overlooked a much older and more
elementary argument given in \cite{hoeffding_probability_1963}, which
only relies on certain convexity properties and applies without change
to the matrix-valued case (see below).

\subsection{Implications}

As a consequence of Theorem~\ref{thm:bernstein}, the analysis in
Section~II.C of \cite{gross_recovering_2009} can be simplified and
improved, by setting the constant $C$ equal to one. The remark at the
end of that section applies. In particular, in the rest of that paper,
one may assume that $\|\Delta_T\|_2 < n^{1/2}\|\Delta_T^\bot\|_2$.
Thus, the conditions on the certificate $Y$ in Section~II.E may be
relaxed to $\|\mathcal{P}_T Y - \operatorname{sgn} \rho\|_2 \leq
\frac{1}{2 n^{1/2}}$. This implies that $l$, the number of iterations
of the ``golfing scheme'', may be reduced to $l=\lceil \log_2(2 n^{1/2}
\sqrt r) \rceil $. The estimates on $|\Omega|$ in Theorems~1, 2, and 3
therefore all improve by a factor of $\frac{\log_2 n^2}{\log_2
n^{1/2}} = 4$.

In \cite{recht_simpler_2009}, Proposition~3.3 becomes superfluous. The
final bounds improve accordingly.

The consequences are more pronounced for an upcoming detailed analysis
\cite{becker__2009} of noise resilience (in the spirit of
\cite{candes_matrix_2009}) of quantum mechanical applications.

The present note makes no statements about approaches which either
rely on the Bernoulli model, or use the non-commutative Kintchine
inequality instead of the operator Chernoff bound
\cite{candes_exact_2009,candes_power_2009,
candes_robust_2009}.  

Finally, note that the ``golfing scheme'' employed in
\cite{gross_quantum_2009,gross_recovering_2009} demands that $l$
independent batches of coefficients be sampled. As a consequence of
Theorem~\ref{thm:bernstein}, every single batch may be assumed to be
drawn without replacement. However, for technical reasons, it is still
necessary that the batches remain independent. This does not
constitute a problem.  Indeed, let $\Omega$ be the set of distinct
coefficients used by the golfing scheme. It is shown that, with high
probability, there exists a ``dual certificate'' in the space spanned
by the basis elements corresponding to the coefficients in $\Omega$.
Since $\Omega$ is just a random subset of cardinality $|\Omega|\leq
m$, the probability that there is a dual certificate in the space
spanned by $m$ distinct random basis elements (obtained from sampling
without replacing) can only be higher. A very similar argument has 
recently been given in \cite{candes_robust_2009}, where the golfing
scheme has been modified to work with the Bernoulli model.

\subsection{Proof}

In this section, we repeat an argument from
\cite{hoeffding_probability_1963} which implies that for all
$\lambda\in\RR$ the inequality  $M_{\bf Y}(\lambda)\leq M_{\bf
X}(\lambda)$ holds. We emphasize that the proof of
\cite{hoeffding_probability_1963} does not need to be modified in
order to apply matrix-valued random variables. However, the
version given below makes some steps explicit which were omitted in
the original paper.

For now, let $C$ be any finite set; let ${\bf X}, {\bf Y}$ be as
above.

The central observation is that one can generate the distribution of
${\bf X}$ by first sampling  ${\bf y}=\langle y_1, \dots, y_m\rangle$
without replacement, and then drawing the $\langle x_1, \dots, x_m
\rangle$ from $\{y_1, \dots, y_m\}$ in a certain (unfortunately not
completely trivial) way.

To make that second step precise, we introduce a random partial
function ${\bf Z}$ from $C^m$ to $C^m$. The domain of $f$ is the set
of vectors ${\bf y}\in C^m$ with pairwise different components ($y_i
\neq y_j$).  Given such a vector ${\bf y}$, we sequentially assign
values to the components $Z_1, \dots, Z_m$ of ${\bf Z(y)}$ by sampling
from $\{y_1, \dots, y_m\}$ according to the following recipe.  At the
$k$th step, let $D_k$ be the subset of $\{y_1, \dots, y_m\}$ of values
which have already been drawn in a previous step. To get $Z_k$: 
\begin{enumerate}
	\item with probability $\frac{|D_k|}{|C|}$ take a random element from
	$D_k$, and

	\item
	with probability $1-\frac{|D_k|}{|C|}$ take a random element from
	the $\{y_1, \dots, y_m\}$ not contained in $D_k$.
\end{enumerate} 
(Here, by a ``random'' element, we mean one sampled uniformly at
random from the indicated set). Then

\begin{lemma}\label{lem:emulate}
	With the definitions above, ${\bf X}$ and ${\bf Z(Y)}$  are
	identically distributed.

	What is more, if $C$ is a subset of a vector space, then
	\begin{equation}\label{eqn:symmetric}
		\EE_{\bf Z}
		\Big[\sum_{i=1}^m Z_i({\bf Y})\Big]
		=\sum_{i=1}^m Y_i.
	\end{equation}
\end{lemma}

\begin{proof}
	Choose $k\in\{1,\dots,m\}$, let ${\bf x}\in C^m$. We compute the
	conditional probability 
	\begin{equation*}
		\operatorname{Pr}\big[
			Z_k({\bf Y})=x_k \,|\,
			Z_1({\bf Y})=x_1, \dots, Z_{k-1}({\bf Y})=x_{k-1}
		 \big].
	\end{equation*}
	If there is a $j<k$ such that $x_k = x_j$, then, according to
	the first rule above, the probability is 
	\begin{equation*}
		\frac{|D_k|}{|C|} \frac1{|D_k|} = \frac{1}{|C|}.
	\end{equation*}
	Otherwise, by the second rule, the probability reads
	\begin{equation*}
		\left(1-\frac{|D_k|}{|C|}\right) \frac1{|C|-|D_k|} = \frac1{|C|}
	\end{equation*}
	as well. Iterating:
	\begin{eqnarray*}
		&&
		\Pr[Z_1({\bf Y})=x_1, \dots, Z_m({\bf Y})=x_m] \\
		&=&
		\Pr[Z_1({\bf Y})=x_1, \dots, Z_{m-1}({\bf Y})=x_{m-1}] 
		\frac1{|C|} \\
		&=&
		\Pr[Z_1({\bf Y})=x_1, \dots, Z_{m-2}({\bf Y})=x_{m-2}]
		\frac1{|C|^2} \\
		&=& \dots = \frac1{|C|^m}.
	\end{eqnarray*}
	This proves the first claim.
	
	We turn to the second statement. The left hand side of
	(\ref{eqn:symmetric}) is manifestly a linear combination of the
	random variables $Y_i$. From the definition of ${\bf Z}$, it is also
	invariant under any permutation $Y_i \mapsto Y_{\pi(i)}$. As a
	linear and symmetric function, it is of the form
	$K \sum_i^m Y_i$ for some constant $K$. To compute $K$, we use the
	fact that the $Y_i$ are identically distributed, so that
	\begin{eqnarray*}
		\EE_{\bf Y} \Big[\EE_{\bf Z}
		\Big[\sum_{i=1}^m Z_i({\bf Y})\Big]\Big]
		&=& m\,\EE[Y_1], \\
		\EE_{\bf Y}\Big[K \sum_{i=1}^m Y_i\Big]&=&K m\,\EE[Y_1].
	\end{eqnarray*}
	Thus $K=1$ and we are done.
\end{proof}

Now let $f$ be a convex function on the convex hull of $C$. Using
Jensen's inequality and
Lemma~\ref{lem:emulate},
\begin{eqnarray*}
	\EE_{\bf X} \Big[ f\big(\sum_{i=1}^m X_i\big) \Big]
	&=&
	\EE_{\bf Y} \EE_{\bf Z} 
	\Big[ f\big(\sum_{i=1}^m Z_i({\bf Y})\big) \Big] \\
	&\geq&
	\EE_{\bf Y}  
	\Big[ f\big(
		\EE_{\bf Z}\big[
		\sum_{i=1}^m Z_i({\bf Y})
		\big]
	\big) \Big]  \\
	&=&
	\EE_{\bf Y}  
	\Big[ f\big(
		\sum_{i=1}^m  Y_i
	\big) \Big].
\end{eqnarray*}

Finally, specialize to the case where $C$ is a finite set of Hermitian
matrices. Since the function $c \mapsto \tr\exp(\lambda c)$ is convex
on the set of Hermitian matrices for all $\lambda\in\RR$, any upper
bound on moment generating functions derived for matrix-valued
sampling with replacing is also valid for sampling without replacing. 

\bibliographystyle{IEEEtran}
\bibliography{compressed}

\begin{thebibliography}{10}
\providecommand{\url}[1]{#1}
\csname url@samestyle\endcsname
\providecommand{\newblock}{\relax}
\providecommand{\bibinfo}[2]{#2}
\providecommand{\BIBentrySTDinterwordspacing}{\spaceskip=0pt\relax}
\providecommand{\BIBentryALTinterwordstretchfactor}{4}
\providecommand{\BIBentryALTinterwordspacing}{\spaceskip=\fontdimen2\font plus
\BIBentryALTinterwordstretchfactor\fontdimen3\font minus
  \fontdimen4\font\relax}
\providecommand{\BIBforeignlanguage}[2]{{%
\expandafter\ifx\csname l@#1\endcsname\relax
\typeout{** WARNING: IEEEtran.bst: No hyphenation pattern has been}%
\typeout{** loaded for the language `#1'. Using the pattern for}%
\typeout{** the default language instead.}%
\else
\language=\csname l@#1\endcsname
\fi
#2}}
\providecommand{\BIBdecl}{\relax}
\BIBdecl

\bibitem{gross_recovering_2009}
\BIBentryALTinterwordspacing
D.~Gross, ``Recovering low-rank matrices from few coefficients in any basis,''
  \emph{preprint}, Oct. 2009. [Online]. Available:
  \url{http://arxiv.org/abs/0910.1879}
\BIBentrySTDinterwordspacing

\bibitem{recht_guaranteed_2007}
\BIBentryALTinterwordspacing
B.~Recht, M.~Fazel, and P.~A. Parrilo, ``Guaranteed {Minimum-Rank} solutions of
  linear matrix equations via nuclear norm minimization,'' \emph{preprint},
  Jun. 2007. [Online]. Available: \url{http://arxiv.org/abs/0706.4138}
\BIBentrySTDinterwordspacing

\bibitem{candes_exact_2009}
E.~Candes and B.~Recht, ``Exact matrix completion via convex optimization,''
  \emph{Foundations of Computational Mathematics}, vol.~9, no.~6, pp. 717--772,
  Dec. 2009.

\bibitem{candes_power_2009}
\BIBentryALTinterwordspacing
E.~J. Candes and T.~Tao, ``The power of convex relaxation: {Near-Optimal}
  matrix completion,'' \emph{preprint}, Mar. 2009. [Online]. Available:
  \url{http://arxiv.org/abs/0903.1476}
\BIBentrySTDinterwordspacing

\bibitem{singer_uniqueness_2009}
\BIBentryALTinterwordspacing
A.~Singer and M.~Cucuringu, ``Uniqueness of {Low-Rank} matrix completion by
  rigidity theory,'' \emph{preprint}, Feb. 2009. [Online]. Available:
  \url{http://arxiv.org/abs/0902.3846}
\BIBentrySTDinterwordspacing

\bibitem{keshavan_matrix_2009}
\BIBentryALTinterwordspacing
R.~H. Keshavan, A.~Montanari, and S.~Oh, ``Matrix completion from a few
  entries,'' \emph{preprint}, 2009. [Online]. Available:
  \url{http://arxiv.org/abs/0901.3150}
\BIBentrySTDinterwordspacing

\bibitem{wright_robust_2009}
\BIBentryALTinterwordspacing
J.~Wright, A.~Ganesh, S.~Rao, and Y.~Ma, ``Robust principal component analysis:
  Exact recovery of corrupted {Low-Rank} matrices,'' \emph{preprint}, May 2009.
  [Online]. Available: \url{http://arxiv.org/abs/0905.0233}
\BIBentrySTDinterwordspacing

\bibitem{gross_quantum_2009}
\BIBentryALTinterwordspacing
D.~Gross, Y.~Liu, S.~T. Flammia, S.~Becker, and J.~Eisert, ``Quantum state
  tomography via compressed sensing,'' \emph{preprint}, Sep. 2009. [Online].
  Available: \url{http://arxiv.org/abs/0909.3304}
\BIBentrySTDinterwordspacing

\bibitem{recht_simpler_2009}
\BIBentryALTinterwordspacing
B.~Recht, ``A simpler approach to matrix completion,'' \emph{preprint}, Oct.
  2009. [Online]. Available: \url{http://arxiv.org/abs/0910.0651}
\BIBentrySTDinterwordspacing

\bibitem{candes_robust_2009}
\BIBentryALTinterwordspacing
E.~Candes, X.~Li, Y.~Ma, and J.~Wright, ``Robust principal component
  analysis?'' \emph{preprint}, 2009. [Online]. Available:
  \url{http://arxiv.org/abs/0912.3599}
\BIBentrySTDinterwordspacing

\bibitem{ahlswede_strong_2002}
R.~Ahlswede and A.~Winter, ``Strong converse for identification via quantum
  channels,'' \emph{{IEEE} Transactions on Information Theory}, vol.~48, no.~3,
  pp. 569--579, 2002.

\bibitem{joag-dev_negative_1983}
K.~{Joag-Dev} and F.~Proschan, ``Negative association of random variables with
  applications,'' \emph{The Annals of Statistics}, vol.~11, no.~1, pp.
  286--295, 1983.

\bibitem{dubhashi_concentration_2009}
D.~P. Dubhashi and A.~Panconesi, \emph{Concentration of Measure for the
  Analysis of Randomized Algorithms}.\hskip 1em plus 0.5em minus 0.4em\relax
  Cambridge University Press, Jun. 2009.

\bibitem{hoeffding_probability_1963}
W.~Hoeffding, ``Probability inequalities for sums of bounded random
  variables,'' \emph{Journal of the American Statistical Association}, vol.~58,
  no. 301, pp. 13--30, Mar. 1963.

\bibitem{becker__2009}
S.~Becker, S.~T. Flammia, D.~Gross, Y.~Liu, and J.~Eisert, 2009, in
  preparation.

\bibitem{candes_matrix_2009}
E.~J. Candes and Y.~Plan, ``Matrix completion with noise,'' \emph{Proceedings
  of the {IEEE}}, 2009.

\end{thebibliography}

\end{document}